\documentclass[11pt]{amsart}
\usepackage{amsmath}
\usepackage{amssymb}
\usepackage{amsthm}
\usepackage{mathrsfs}
\usepackage{comment}
\usepackage{hyperref}
\usepackage{algorithm,algpseudocode}
\algblock{Input}{EndInput}
\algnotext{EndInput}
\algblock{Output}{EndOutput}
\algnotext{EndOutput}

\usepackage[all,cmtip]{xy}\usepackage{xcolor}
\usepackage{enumerate}
\usepackage{bm}
\usepackage{dsfont}
\usepackage{mathtools}
\usepackage[cal=euler]{mathalfa}
\usepackage[top=1in, bottom=1.25in, left=1.25in, right=1.25in]{geometry}
\usepackage{parskip}

\hypersetup{colorlinks=true,linkcolor=magenta,citecolor=blue}

\DeclareMathAlphabet{\mathpzc}{OT1}{pzc}{m}{it}

\theoremstyle{plain}
\newtheorem{theorem}{Theorem}[section]
\newtheorem*{theorem*}{Theorem}

\newtheorem{lemma}[theorem]{Lemma}
\newtheorem*{claim*}{Claim}
\newtheorem{proposition}[theorem]{Proposition}

\newtheorem{corollary}[theorem]{Corollary}

\theoremstyle{definition}
\newtheorem{definition}[theorem]{Definition}

\newtheorem{example}[theorem]{Example}

\numberwithin{equation}{section}
\numberwithin{figure}{section}

\newcommand{\ZZ}{\mathbb{Z}}
\newcommand{\NN}{\mathbb{N}}
\newcommand{\FF}{\mathbb{F}}
\newcommand{\extK}{K} %some extension field

\newcommand{\Sone}[1]{\Stsigma/\Stsigma (t^{\de}-#1)}

\newcommand{\ring}{S}
\newcommand{\field}{K}

\newcommand{\Stsigma}{S[t;\sigma]}
\newcommand{\Kone}[1]{K[t;\sigma]/K[t;\sigma](t^{\de}-#1)}

\DeclareMathOperator{\id}{id}

\newcommand{\os}{{\mathbf{m}}}  %temporary symbol for the order of $\sigma$
\newcommand{\de}{{\mathbf{n}}}  %temporary symbol for the degree of $f$
\newcommand{\N}{\mathbb{N}}

\newcommand{\Z}{\mathbb{Z}} %integers

\newcommand{\ignore}[1]{}

\DeclareMathOperator{\Gal}{Gal}
\DeclareMathOperator{\Aut}{Aut}

\begin{document}
\title{The first  tight classification of skew-constacyclic codes over finite fields}

\author{Monica Nevins}
\address{Department of Mathematics and Statistics, University of Ottawa, Ottawa, Canada K1N 6N5}
\email{mnevins@uottawa.ca}
%\thanks{The first author's research is supported by NSERC Discovery Grant RGPIN-2025-05630.}

\author{Susanne Pumpl\"un}
\address{School of Mathematical Sciences, University of Nottingham,
Nottingham NG7 2RD, United Kingdom}
\email{Susanne.Pumpluen@nottingham.ac.uk}

\keywords{Skew polycyclic codes, skew polynomials, skew constacyclic codes, nonassociative algebras, isometries}

\subjclass[2020]{94B15, 11T71, 17A99}

\date{\today}

\begin{abstract}
We parametrize the isometry and equivalence classes of skew constacyclic codes  over a finite field by classifying the corresponding classes of their ambient rings, and present algorithms for the parametrizations. We achieve a tight classification by taking all possible Hamming-weight preserving isomorphisms between their ambient Petit rings into account.  We also count the number of equivalence classes %(and Chen-equivalence classes)
of these rings.  We present many examples where isometry is a strictly stronger relation than equivalence, that is, codes that are isometric but not equivalent.
\end{abstract}

\maketitle

\section{Introduction}
The class of error-correcting codes known as skew constacyclic codes, first introduced in \cite{BoucherGeiselmannUlmer2007}, has risen in prominence over the past several years, due to their useful algebraic structure and efficient coding and decoding algorithms.  An essential aspect to optimizing the choices of these codes is to establish their classification up isomorphisms that preserve the main performance parameters of the code (such as length, dimension, minimum Hamming distance), that is, up to Hamming-weight preserving isomorphisms which are called \emph{isometries}.  To know the isometry class of a code means we can optimize searches for optimal codes and  extensive work has been done recently towards classifying  skew constacyclic codes up to various isometry relations \cite{BoulanouarBatoulBoucher2021, BLMS2024, OuazzouNajmeddineAydin2025, LobilloMunoz2025, NevinsPumpluen2025submitted}.

However, most current results on skew constacyclic codes routinely overestimate the number of distinct isometry classes of skew constacyclic codes, as they do not take into account all the possible Hamming-weight and dimension  preserving isomorphisms that exist between their (potentially nonassociative) ambient rings.  In this paper, we correct this oversight and show that there can exist higher-degree isometries  between skew constacyclic codes, as well as isometries that act by automorphisms on the base field, that are not realized by the more classic notion of what we term \emph{Chen equivalence} that is prevalent in the literature. For a taste of these distinctions, see Example \ref{Eg:m=135n=9q=16}.

To do so, we exploit a powerful algebraic tool, which is that
skew $(\sigma,a)$-constacyclic codes of length $\de$ are in one-one correspondence with the principal left ideals of their ambient  Petit rings $K[t,\sigma]/K[t,\sigma](t^\de-a)$.  Here, $K$ is a field, $a\in K^\times$ and $\sigma\in \Aut(K)$. These rings may be nonassociative and their definition is subtle (see Section~\ref{subsec:nonass}).  Adopting this perspective allows us to address the classification of these codes uniformly across all choices of length $\de$  and order $\os$ of $\sigma$.  Briefly, we let ${\bf C}_a$ be the class of all skew $(\sigma,a)$-constacyclic codes of length $\de$, that is, with ambient ring $K[t;\sigma]/K[t;\sigma](t^\de-a)$.

In prior work, the second author established the full class of Hamming-weight preserving isomorphisms  between the ambient rings of the form $\Kone{a}$, as $a$ varies over $K^\times$ \cite{Pumpluen2025}.  Given such an isometry $G_{\tau,\alpha,k}: \Kone{a}\to \Kone{b}$, its restriction to the left principal ideals defines a %dimension and
parameter-preserving \emph{isometry} between the corresponding skew $(\sigma,a)$-constacyclic codes and the skew $(\sigma,b)$-constacyclic code of the same length $\de$.  Consequently, we define this as an \emph{isometry} between the classes $\mathbf{C}_a$ and $\mathbf{C}_b$.

 These isometries have a distinct form. Every isomorphism between the  ambient rings $K[t,\sigma]/K[t,\sigma](t^\de-a)$ and $K[t,\sigma]/K[t,\sigma](t^\de-b)$ of two skew constacyclic codes of length $\de$ over  $K$  that preserves the Hamming weight is necessarily \emph{monomial} of some degree $k$.  More precisely, an isomorphism $G_{\tau,\alpha,k}$ is determined by some $\tau\in \Aut(K), \alpha\in K^\times$ and $k\in \NN$, and it is generated by mapping $t$ to $\alpha t^k$ and $a\in K$ to $\tau(a)$.   Note that these isomorphisms are in one one correspondence with the $\tau$-semilinear bijective Hamming-weight preserving isomorphism between the modules $K[t,\sigma]/K[t,\sigma](t^\de-a)$ and $K[t,\sigma]/K[t,\sigma](t^\de-b)$.

 The isomorphisms $G_{\tau,\alpha,k}$ exhaust the collection of Hamming-weight and dimension preserving isomorphisms between the corresponding codes and thus henceforth we refer to $G_{\tau,\alpha,k}$ as an \emph{isometry}.  Isometries specialize to the case commonly used in the literature: if $\tau=\id$ and $k=1$ then this is what we term \emph{Chen equivalence} (introduced in \cite{Chen2012} for constacyclic codes and used for instance in \cite{LobilloMunoz2025, OuazzouHorlemannAydin2025, OuazzouNajmeddineAydin2025}, where it is sometimes called $(\de,\sigma)$-equivalence).  When one allows $\tau$ to vary but restricts to $k=1$,  we write  $G_{\tau,\alpha}$ and this is what we term an \emph{equivalence}.   In \cite{NevinsPumpluen2025submitted} we proved that all isometries are equivalences in under certain hypotheses (see  Theorem~\ref{T:isometryandequivalencecoincide} here).

 We know already that for many choices of $\de$ and $\os$, the notions of isometry and equivalence coincide not just  for skew constacyclic codes  (\cite[Corollary 4.5]{NevinsPumpluen2025submitted}), but even for general skew polycyclic codes \cite{Pumpluen2025}.
 However, despite articles such as \cite{Bajalanetal2026} which argue the contrary), it is \emph{not} true that all isometries are equivalences in general, and it is a defect of the literature to assume that all isometries fix the fixed field of $\sigma$ or have degree $1$.

%Importantly, note that there is  a ``selection bias'' in the current literature on skew polycyclic codes in the sense that any isometry between two skew polycyclic codes is either assumed to restrict to the identity on the base field (that is, $\tau=\id$) and/or must have $k=1$.

Therefore the strongest classification one can obtain for fixed $\sigma$ is by employing all the isomorphisms $G_{\tau,\alpha,k}$. In this case we talk about \emph{isometry classes} of ambient rings, and about \emph{isometry classes} of skew constacyclic codes.  Similarly, we refer to \emph{Chen equivalence classes} (when $\tau=\id$ and $k=1$) and \emph{equivalence classes} (when $k=1$ but $\tau$ can vary).  This  terminology, and this all-encompassing way of classifying codes, was presented for the first time in \cite{Pumpluen2025}. In prior work \cite{NevinsPumpluen2025submitted}, we counted the skew constacyclic codes  up to Chen equivalence, correcting also an error in \cite{OuazzouNajmeddineAydin2025}.  %(Did they also do Chen isometry?  If you are thinking of , wasn't their count wrong? Is there someone else we should cite? (a paragraph commented out below as it hopefully is included above now)} \textcolor{blue}{only we got the right numbers, nobody did Chen isometry or isometry in our sense. I recall counting the chen equivalence and isometries in our last paper, but no parametrization of the Chen equivalence classes, will check if Ouazzou did that}

%In this paper, we count and parametrize isometry classes.  To respect the current bias in the literature, in this paper we also parametrize and count the  that is, up to isometries $G_{\id,\alpha}$ and $G_{\tau,\alpha}$, respectively.

%\textcolor{orange}{There are even papers proving $k=1$ when the ambient algebra of a certain skew polycyclic code (a tricyclic code) is associative, e.g. \cite{Bajalanetal2026}. }\textcolor{red}{I do not understand the last sentence in orange.  Do you mean "purport to prove" (as in, they are wrong) or do you mean to say "For some combinations of $K,m,n$ the only isometries are of degree $k=1$, as has been proven in (cite them and us)?} \textcolor{blue}{meaning they are wrong!! how to write that?}

 %So far, skew constacyclic codes have only been counted and classified up to what we will call Chen equivalence and Chen isometry. Two codes are called \emph{Chen equivalent}, if there exists a ring isomorphism $G_{\id,\alpha}$ between their ambient rings that maps the set of code words of one skew polycyclic code to the other, after they have been identified with a set of skew polynomials.   Such an isomorphism $G_{\id,\alpha}$ restricts to the identity map on $K$ and is monomial of degree one; it maps $t$ to $\alpha t$. Analogously, two codes are called \emph{Chen isometric}, if there exists a ring isomorphism $G_{\id,\alpha,k}$ between their ambient rings that maps the set of code words of one skew polycyclic code to the other, again after they have been canonically identified with a set of skew polynomials.

In this paper, we focus on finite base fields.  After sharing some background in Section~\ref{sec:notation}, we classify skew constacyclic codes by both equivalence (in Section~\ref{sec:countingequivclasses}) and isometry (in Section~\ref{sec:isoclasses}).   More specifically, we count the classes of their ambient rings and provide algorithms for generating a set of representatives for the classes  ${\bf C}_a$.  We present many examples, expanding upon some first examples in \cite{Pumpluen2025}.  We go on in Section~\ref{sec:examples} to generate families of examples of isometric but non-equivalent codes.  We also completely characterize the circumstances under which isometry and equivalence do coincide, and the circumstances under which all skew constacyclic codes are isometric with skew cyclic codes (that is, skew $(\sigma,1)$-constacyclic codes). %,  where two codes are called \emph{equivalent}
%or  also \emph{1-equivalent} for emphasis, if there exist a monomial isomorphism $G_{\tau,\alpha}$ of degree one between their ambient rings that yields the one-one correspondence between the codes.

%We will both count the number of equivalence and isometry classes and also give algorithms to find representatives for these classes in Sections \ref{sec:countingequivclasses} and \ref{sec:isoclasses}.

%When $m=2$ or $3$, or when $m\geq 5$ is prime and  $F$ contains a primitive $\de$th root of unity, then under the assumption that  $\gcd\left(q^\os-1, \frac{q^\de-1}{q-1}\right)=1$, we show that all skew $(\sigma,a)$-constacyclic codes are Chen equivalent to a skew cyclic code. In particular, all skew $(\sigma,a)$-constacyclic codes are isometric to a skew cyclic code (Proposition \ref{prop:skewcyclic}).

%Expanding upon some first examples in \cite{Pumpluen2025}, we establish a rigorous counting mechanism that eliminates redundancies that occur when only working with Chen-equivalence and Chen isometry (which use only the isomorphisms $G_{\id,\alpha,k}$).

Let us present our main results in greater detail.  The case of equivalence classes (and Chen equivalence classes) is treated in Section~\ref{sec:countingequivclasses}.  One key theorem is the following, which counts the equivalence classes.

\begin{theorem*}[Theorem~\ref{them:countingskewconstacodes}]
Suppose $F=\FF_{q}$ and $K=\FF_{q^\os}$ are finite fields of characteristic $p$ and that $\sigma$ generates $\Gal(K/F)$.  For each divisor $d$ of $s=\gcd(q^\os-1, (q^\de-1)/(q-1))$, write $\varphi(d)=|(\Z/d\Z)^\times|$ for the Euler phi function of $d$ and for any $\ell$ such that $\gcd(\ell,d)=1$, write $o(\ell)_d$ for the order of the image of $\ell$ in $(\Z/d\Z)^\times$.
  Then the number of equivalence classes of ambient Petit algebras $\Kone{a}$  (that is, the number of  equivalence classes  ${\bf C}_a$) as $a$ varies over $K^\times$,
  is
    $$
   N= \sum_{d|s}\frac{\varphi(d)}{o(p)_d}.
    $$
 %When  $\os \nmid\de$, the number of isometry classes  of skew constacyclic codes of length $\de$  over $\mathbb{F}_{q^\os}$ is also $\os$.
Moreover, when $\os \nmid\de$ this is also the number of \emph{isometry} classes  of these rings, i.e. the number of  \emph{isometry} classes  ${\bf C}_a$.
\end{theorem*}

 We go on to give an algorithm for computing representatives of all equivalence classes (Algorithm~\ref{Algorithm:1}).

 In Section \ref{sec:isoclasses}, we turn to isometry classes.
By Theorem~\ref{T:isometryandequivalencecoincide},
 isometry and equivalence coincide for all proper nonassociative ambient Petit algebras, that is, whenever $\os\nmid \de$ and $a \notin F^\times$.  Therefore we only have to characterize isometry when the ambient algebras are associative.

 Our main theorem in this section is the following.  Write $(F^\times)^z = \{c^z\mid c\in F^\times\}$ for the subgroup of $z$th powers in $F^\times$. %Algorithm ~\ref{Algorithm:2} uses the following main Theorem \ref{T:isometrydifferentfromequivalence}.

\begin{theorem*}[Theorem \ref{T:isometrydifferentfromequivalence}]
Let $K=\FF_{q^\os}$ and $F=\FF_q$ be finite fields of characteristic $p$.
Suppose $\os\mid\de$ and set $s_0=\gcd(\de/\os,q-1)$.  Then two \emph{inequivalent} ambient Petit rings $K[t;\sigma]/K[t;\sigma](t^\de-a)$ and $K[t;\sigma]/K[t;\sigma](t^\de-b)$ are isometric over $\FF_p$  if and only if
     $a,b\in F^\times$ and
there exists $\tau \in \Gal(K/\FF_p)$ such that $$\tau(a)\in b^k (F^\times)^{s_0}$$
for some $k\in \NN$ satisfying $$
2\leq k <\de, \quad k\equiv 1 \mod \os, \quad \text{and} \quad \gcd(k,\de)=1.
$$
\end{theorem*}

Building on Algorithm~\ref{Algorithm:1}, we produce an algorithm for computing representatives of all isometry classes (Algorithm~\ref{Algorithm:2}).

In Section~\ref{sec:examples}, we give examples of isometric but non-equivalent ambient Petit algebras, as well as examples of skew constacyclic codes that are thus isometric but not equivalent.  The key result is the following.

\begin{theorem*}[Proposition \ref{prop:proper equivalences}]
    Suppose $\os\mid\de$ and set $S = \{ 1+r\os \mid 1\leq r < \de/\os, \gcd(1+r\os,\de)=1\}$.  There exists $a\in F^\times$ such that $\Kone{a}$ and $\Kone{a^k}$ are isometric but not equivalent if %and $k\in \NN$ such that $G_{\id,1,k}:\Kone{a^k}\to \Kone{a}$ is an isometry, but   but no equivalence $G_{\tau,\alpha}$ between them if
    all the following hold:
    \begin{itemize}
        \item $s_0=\gcd(\de/\os,q-1)>1$;
        \item the subgroup $\mathscr{P}$ of $(\ZZ/s_0\ZZ)^\times$ generated by $p$ is proper (that is, $p$ is not a primitive element);
        \item there exists $k\in S$ %$k\in [2,m)$ such that $k\equiv 1\mod n$, $\gcd(k,m)=1$ and
        such that the image of $k$ in $(\ZZ/s_0\ZZ)^\times$ is not in $\mathscr{P}$.
    \end{itemize}
    Note that the last condition implies $S$ is nonempty and that $s_0\nmid n$.
\end{theorem*}

 Our results can be used to identify equivalent codes
  and to optimize searches for good codes in an even larger setting. For previous extensive work on classifying  general
skew polycyclic codes (SPC codes, also called skew $(\sigma,\delta)$-polycyclic codes or $(\sigma,\delta)$-codes) via different equivalence relations which are  related to our general approach initiated in \cite{Pumpluen2025}, we refer the reader to for example
\cite{BoulanouarBatoulBoucher2021}. %The term $(\de,\sigma)$-equivalence in  \cite{OuazzouNajmeddineAydin2025} corresponds to  Chen equivalence,
% and the term $(\de,\sigma)$-isometry  to what we will call Chen isometry (employing the isometries $G_{id,\alpha,k}$).
In upcoming work, we extend our results to skew polycyclic codes over finite fields and chain rings.

\section{A collection of terminology and previous results}\label{sec:notation}

To describe our results, let us first introduce some notation. Let $K/F$ be a finite Galois field extension. %Given a Galois extension $K/F$, a generator $\sigma\in \Gal(K/F)$ and an element $a\in K^\times$, we can define the class of skew $(\sigma,a)$-constacyclic codes of length $\de$ to be those subspaces of $K^{\de}$ that are stable under the $(\sigma,a)$-constacyclic shift operator
%$$
%T_{\sigma,a}(b_0,b_1,\cdots,b_{\de-1}) = (\sigma(b_{\de-1})a, \sigma(b_{0}), \sigma(b_1) , \cdots, \sigma(b_{\de-2})).
%$$
Denote by $K[t;\sigma]$  the skew polynomial ring, where $\sigma$ generates $\Gal(K/F)$ and $tb=\sigma(b)t$ for all $b\in K$.
Then a skew $(\sigma,a)$-constacyclic code of length $\de$ is generated by a monic right divisor $g(t)$ of $t^\de-a\in K[t;\sigma]$ and can be identified with a left principal ideal in the  ambient Petit ring $K[t;\sigma]/K[t;\sigma](t^\de-a)$ that is generated by $g(t)$. The  ambient Petit ring $K[t;\sigma]/K[t;\sigma](t^\de-a)$,  which we present in the next  Section~\ref{subsec:nonass}, is a generally nonassociative ring, defined on the set of skew polynomials of degree less than $\de$.  We define skew $(\sigma,a)$-constacyclic  codes and isometries more precisely in Section~\ref{subsec:codes}, and present some known results in Section~\ref{subsec:results}.

%.  Identifying this quotient space with the set of skew polynomials of degree less than $\de$ allows one to define a generally nonassociative algebra structure, as we describe in detail in Section~\ref{subsec:nonass}.

\subsection{Petit rings} \label{subsec:nonass}

Let $f \in R=K[t;\sigma]$ have degree $\de>1$ and let ${\rm mod}_r f$ denote the remainder of right division by $f$. Let $\sigma$ have finite order $\os$.
 The vector space $R_\de=\{g\in K[t;\sigma]\,|\, {\rm deg}(g)<\de\}$ together with the multiplication
 $g\circ h=gh \,\,{\rm mod}_r f $, where the right hand side is the remainder obtained after dividing $gh$ on the right by $f$,
 is a unital nonassociative ring (respectively, a unital algebra  over $F$ of dimension $\de \os$)
  denoted by $R/Rf$ or $\mathbb{S}_f$, and called a \emph{Petit ring} (respectively, a \emph{Petit algebra}), and a \emph{proper Petit algebra or ring} when $\mathbb{S}_f$ is not associative.

  When ${\rm deg}(g)+{\rm deg}(h)<\de$, the multiplication of $g$ and $h$ in $\mathbb{S}_f$ is thus the same as the multiplication of $g$ and $h$ in $R$. The algebra $\mathbb{S}_f$ is associative if and only if $f$ is \emph{two-sided}, meaning $Rf$ is a two-sided ideal.
In that case, $\mathbb{S}_f$ is the usual quotient algebra  $R/\langle f\rangle$ \cite[(7),(9), (10)]{Petit1967}.

% If  $\mathbb{S}_f$ is a proper Petit algebra, then $\mathbb{S}_f$ has center $F$,
%${\rm Nuc}_l(\mathbb{S}_f)={\rm Nuc}_m(\mathbb{S}_f)=K$and$${\rm Nuc}_r(\mathbb{S}_f)=\{g\in \mathbb{S}_f\,|\, fg\in Rf\}$$ is the eigenspace of $f$ \cite{Petit1967}.

 In this paper, $f(t)=t^\de-a\in K[t;\sigma]$, $a\in K^\times$, since the ring $K[t;\sigma]/K[t;\sigma](t^\de-a)$ is the ambient ring of a skew $(\sigma,a)$-constacyclic code when  $f(t)=t^\de-a\in K[t;\sigma]$ is reducible.
 The algebra $R/R(t^\de-a)$
 is  associative if and only if $\de\mid \os$  and $a\in F^\times$ (e.g. cf. \cite{Petit1967}).  Thus we know that  $f(t)=t^\de-a\in K[t;\sigma]$ is not two-sided, if and only if either $\de\nmid \os$ or  $\de\mid \os$ but $a\in \extK\smallsetminus F$.

When $f(t)=t^\de-a$ is reducible has to be decided on a one-by-one basis. We know that if $\alpha\in K^\times$ and $t-\alpha$ right divides $t^\de-a$ then this implies that
$$N_\de^\sigma(\alpha)=a$$
 \cite{Jacobson1996}.
When $\de$ is prime we know more.

\begin{theorem} \label{T:divalg} \cite[Theorem 3.11]{BrownPhD2018}
Suppose that $\de=2$ or $3$, or that $\de\geq 5$ is a prime and that $F$ contains a primitive $\de$th root of unity.
Then  $R/R(t^\de-a)$ is a division algebra if and only if
$$\sigma^{\de-1}(z)\cdots \sigma(z) z\not=a$$
for all $z\in K$.
\end{theorem}

For more background on Petit algebras we refer the reader to \cite{Pumpluen2025, Petit1967}.

\subsection{Different code equivalences that refine isometries}\label{subsec:codes}

Let $K$ be a field, $\\sigma\in {\rm Aut}(K)$, and $F$ the fixed field
of $\sigma$.

A \emph{linear code $C$ of length $\de$ over $K$} is a  sub vector space of the $K$-vector space $K^\de$.  A linear code $C\subset K^{\de}$ is called a  \emph{skew  constacyclic code} or, more precisely, a \emph{skew  $(\sigma, a)$-constacyclic code}, if for each codeword  $(c_0,c_1,\ldots, c_{\de-1})$ of $C$,
$$(\sigma(c_{\de-1})a,\sigma(c_0),\dots,\sigma(c_{\de-2}))\in  C.$$
When $a=1$ we call $C$ a \emph{skew cyclic code}, and when $\sigma=id$, $C$
 is called a \emph{constacyclic code}. Among the constacyclic codes, those with $a=1$ are called \emph{cyclic} codes and those with $a=-1$ \emph{negacyclic codes}.

 The ambient ring of a skew  $(\sigma, a)$-constacyclic code is the (potentially nonassociative) Petit ring $K[t;\sigma]/K[t;\sigma](t^\de-a)$. More precisely, every  skew  $(\sigma, a)$-constacyclic code $C$ can be identified with  a set of skew polynomials which is a left principal ideal in its ambient ring $K[t;\sigma]/K[t;\sigma](t^\de-a)$. This left principal ideal is generated by some monic right divisor $g(t) $ of the skew polynomial $t^\de-a$. We will use the bijective map $\Phi:K^{\de}\rightarrow K[t;\sigma]/K[t;\sigma](t^\de-a)$,
$$(c_0,c_1,\dots,c_{\de-1})\mapsto \sum_{i=0}^{\de-1}c_it^i,$$
to ``move'' between vectors  of length $\de$  and skew polynomials of degree ${\de}-1$.
%We denote the set of skew-polynomials $c(t)=\sum_{i=0}^{\de-1}c_it^i$ associated to the codewords $(c_0,\dots,c_{\de-1})\in C$ by $C(t)$.

Two skew polycyclic codes of the same length are usually called  isometric, if there is some Hamming-weight preserving bijective map that maps  one skew polycyclic code one on one to the other one, and  that preserves their dimension and  their length. The canonical maps to define isometry are the Hamming-weight preserving, ring isomorphisms of ambient rings. These Hamming-weight preserving ring isomorphisms between the ambient rings of skew constacyclic codes all are of the type $G_{\tau,\alpha,k}$ and preserve the dimension of the codes \cite[Section 2]{Pumpluen2025}.

This inspired the following definition of isometry and also makes it the most general possible in this setting.

\begin{definition}\label{D:Gtaualpha}
      Suppose $\tau\in \Aut(K)$, $\alpha \in K^\times$, and $k\in \mathbb{N}$.
    If there exists a  ring isomorphism
    $$
    G: K[t;\sigma]/K[t;\sigma](t^\de-a) \to K[t;\sigma]/K[t;\sigma](t^\de-b)
    $$
    defined via $G|_{K}=\tau$ and $G(t)=\alpha t^k$,
    that is,
\begin{equation}\label{E:formulaG(t)2}
G \left( \sum_{i=0}^{\de-1}c_it^i \right) = \sum_{i=0}^{\de-1}\tau(c_i) (\alpha t^k)^i.
\end{equation}
     then we call $G$ an \emph{isometry}, or a \emph{monomial isomorphism  of degree $k$} if we want to clarify the degree $k$, and write $G=G_{\tau,\alpha,k}$.
      The ring isomorphism $G_{\tau,\alpha,k}$%:K[t;\sigma]/K[t;\sigma](t^\de-a)\rightarrow K[t;\sigma]/K[t;\sigma](t^\de-b)$
  is an $\tilde{F}$-algebra isomorphism, for every subfield $\tilde{F}\subset {\rm Fix}(\tau)\cap F$.
\end{definition}

      When $k=1$, we write $G_{\tau,\alpha}$ in place of $G_{\tau,\alpha,1}$ and call $G_{\tau,\alpha}$ an \emph{equivalence}  or a \emph{monomial isomorphism  of degree one}.

When $\tau=\id$ then a ring isomorphism of the type
    $G_{\id,\alpha,k}$ %$:K[t;\sigma]/K[t;\sigma](t^\de-a)\rightarrow K[t;\sigma]/K[t;\sigma](t^\de-b)      $
    is called a  \emph{Chen isometry}, and $G_{\id,\alpha}$ is called a \emph{Chen equivalence}.

Define $N_i^\sigma(a)=\sigma^{i-1}(a)\cdots \sigma(a)a$ for all positive integers $i$ and $N_0^\sigma(a)=1$. When $\sigma$ has order $\os$, then  $N_{\os}^\sigma$ is simply the norm  $N_{K/F}$ of the field extension $K/F$.
Note that $\beta \tau(N_i^\tau(\beta))=N_{i+1}^\tau(\beta)$ and thus
\begin{equation}\label{E:Normrelation}
N_{i+j}^\tau(\beta)= N_{i}^\tau(\beta) \cdot \tau^{i}(N_j^\tau(\beta))
\end{equation}
for all $i,j\geq 0$ and automorphisms $\tau$ of $K$.
Then for example with $k=1$ we have
$$G_{\tau,\alpha}(\sum_{i=0}^{\de-1}d_it^i)=\sum_{i=0}^{\de-1}\tau(d_i) N^\sigma_{i}(\alpha) t^i.$$

    % If $\de$ is not clear from our context ($\de$ corresponds to the length of the code), we will sometimes use the terms \emph{$\de$-isometry}, \emph{$\de$-equivalence}, \emph{$\de$-Chen isometry},  and \emph{$\de$-Chen equivalence} instead.

We now introduce the notation that will be in force throughout this paper.  Let ${\bf C}_a$ be the class of all skew $(\sigma,a)$-constacyclic codes of length $\de$, that is, with ambient ring $K[t;\sigma]/K[t;\sigma](t^\de-a)$.

\begin{definition}\label{Definition:isometryChenisometryequivalenceChenequivalence}
The classes ${\bf C}_a$  and ${\bf C}_b$  are called
\begin{enumerate}[(i)]
\item  \emph{isometric}, if there exist $\tau \in \Aut(K), \alpha \in K^\times$ and $k\in \NN$ yielding an isometry
$G_{\tau,\alpha,k}:K[t;\sigma]/K[t;\sigma](t^\de-a)\rightarrow K[t;\sigma]/K[t;\sigma](t^\de-b),$
 \item \emph{Chen isometric}, if there exist $\alpha \in K^\times$ and $k\in \NN$ yielding an isometry
$G_{\id,\alpha,k}:K[t;\sigma]/K[t;\sigma](t^\de-a)\rightarrow K[t;\sigma]/K[t;\sigma](t^\de-b),$
\item  \emph{equivalent},  if there exist $\tau\in \Aut(K)$ and $\alpha\in K^\times$ yielding an equivalence
$G_{\tau,\alpha}:K[t;\sigma]/K[t;\sigma](t^\de-a)\rightarrow K[t;\sigma]/K[t;\sigma](t^\de-b)$,
\item \emph{Chen equivalent},  if there exists  a Chen equivalence
$G_{\id,\alpha}:K[t;\sigma]/K[t;\sigma](t^\de-a)\rightarrow K[t;\sigma]/K[t;\sigma](t^\de-b)$ for some $\alpha \in K^\times$.
\end{enumerate}
% If the length of the codes is not clear from the context, we will use the terminology $\de$-\emph{isometric}, $\de$-\emph{Chen isometric} etc. instead.
\end{definition}

Isometric and equivalent  classes of codes ${\bf C}_a$  and ${\bf C}_b$ correspond to skew constacyclic codes with the same
parameters, thus have the same performance.   That is, if two classes ${\bf C}_a$ and ${\bf C}_b$ are   isometric via some monomial isomorphism $G_{\tau,\alpha,k}$ of degree $k$, then the skew constacyclic  codes in ${\bf C}_a$ correspond bijectively to
  skew constacyclic codes in ${\bf C}_b$ having the same Hamming distance, dimension and length \cite{Pumpluen2025}.

The following was proven in \cite[Corollary 6.2]{NevinsPumpluen2025submitted}.

\begin{theorem}\label{T:isometryandequivalencecoincide}
    When  $\os \nmid \de$, all  isomorphisms between two (automatically proper nonassociative) Petit rings  $K[t;\sigma]/K[t;\sigma](t^\de-a)$ and $K[t;\sigma]/K[t;\sigma](t^\de-b)$  equivalences, that is, of the form $G=G_{\tau,\alpha}$.  Consequently, in this case the notions of isometry and equivalence coincide for the classes ${\bf C}_a$  and ${\bf C}_b$.
\end{theorem}

In what follows, we classify isometries between the Petit rings $\Kone{a}$ and this therefore gives isometries between the classes of codes in the above sense.

\ignore{%%%%%%%%%%%%%%%%%%%%%%%%%%%%
\textcolor{red}{To be discussed: where do the following bits belong and do we need the equivalence on the elements?  There's a lot of repetition and we're using many equivalent phrases (monomial, $1$-equiv etc).  Do we use $a\cong_{Chen}b$ that way later?}

\color{purple}

Inspired by \cite{OuazzouNajmeddineAydin2025, OuazzouHorlemannAydin2025} we look at equivalence relations on  $K^\times$ using the following equivalence relations:

\begin{definition} \label{def:equ}

We call $a$ \emph{isometric} to $b$, written $a\cong_\de b,$
if there exists an isometry $G_{\tau,\alpha,k}:K[t;\sigma]/K[t;\sigma](t^\de-a)\rightarrow K[t;\sigma]/K[t;\sigma](t^\de-b),$
and \emph{Chen isometric}, if $\tau=id$.
\\
We call $a$ \emph{equivalent} to $b$, written
$a\sim_\de b,$
if there exists an equivalence $G_{\tau,\alpha}:K[t;\sigma]/K[t;\sigma](t^\de-a)\rightarrow K[t;\sigma]/K[t;\sigma](t^\de-b)$,
and \emph{Chen equivalent}, written
$a\sim_{Chen, \de} b$
if $\tau=id$.
\end{definition}

The above equivalence relations  partition $(K^\times)^\de$ into different equivalence classes.
Moreover,
$a\cong_\de b$
if and only if ${\bf C}_a$  and ${\bf C}_b$ are isometric,
$a\cong_{Chen, \de} b$
if and only if ${\bf C}_a$  and ${\bf C}_b$ are Chen isometric,
$a\sim_\de b$
if and only if ${\bf C}_a$  and ${\bf C}_b$ are equivalent, and
$a\sim_{Chen, \de} b$
if and only if ${\bf C}_a$  and ${\bf C}_b$ are Chen equivalent.

\color{black}
}%%%%%%%%%%%%%%%%%%%%

 \subsection{A collection of known results}\label{subsec:results}

  We want to classify the ambient Petit rings $K[t;\sigma]/K[t;\sigma](t^\de-a)$
 up to isometries,  as $a$ runs over $K^\times$. In order to do so,
  we  need to describe  the Hamming-weight preserving ring isomorphisms between two  rings $K[t;\sigma]/K[t;\sigma](t^\de-a)$ and $K[t;\sigma]/K[t;\sigma](t^\de-b)$. Recall that
$R/R(t^\de-a)$ is associative if and only if $\os\mid \de$ and $a\in F$.

We first fix a choice of cyclic field extension $K/F$ of degree $\os$.
%Note that when  $\os \nmid \de$ or $\os=\de$, we may  fix a choice of generator $\sigma$ of $\Gal(K/F)$  as well (Theorem~\ref{thm:main1}).

  % $n \nmid m$ or $n\geq m-1$,
  All Hamming-weight preserving ring  isomorphisms between two proper nonassociative ambient rings  $K[t;\sigma]/K[t;\sigma](t^\de-a)$ and $K[t;\sigma]/K[t;\sigma](t^\de-b)$  are monomial of degree one, and of the form $G=G_{\tau,\alpha}$ (\emph{cf.} \cite[Corollary 6.2]{NevinsPumpluen2025submitted}, \cite[Theorem 4.2]{Pumpluen2025} or \cite[Theorem 2.3]{NevinsPumpluen2025submitted}).

 %We will now count the isomorphisms classes between proper nonassociative algebras of the form $K[t;\sigma]/K[t;\sigma](t^\de-a)$ where $a,b\in K\setminus F$ that are given by the isomorphisms $G=G_{\tau,\alpha}$. These are all possible classes when  $n \nmid m$ or  $n \geq m-1$.

 We do not assume that $t^\de-a\in K[t;\sigma]$ is reducible, if it is irreducible we will count the trivial skew constacyclic code obtained from $K[t;\sigma]/K[t;\sigma](t^\de-a)$.

Let us call the isomorphisms classes with respect to $G_{\tau,\alpha}$ the \emph{degree one
isomorphism classes} (or simply the  \emph{equivalence classes} as in \cite{Pumpluen2025}).

% In a large number of cases every Hamming-weight preserving isomorphism between proper nonassociative ambient rings of skew constacyclic codesis  monomial of degree one:when  $a,b\in K^\times$ and $n\nmid m$, then any nonzero homomorphism $G:K[t,\sigma]/K[t,\sigma](t^\de-a)\toK[t,\sigma]/K[t,\sigma](t^\de-b)$  whose restriction to $K$ is given by some automorphism $\tau$ commuting with $\sigma$, must be monomial of degree one, i.e. of the form  $G_{\tau,\alpha}$ \cite[Corollary 6.2]{NevinsPumpluen2025submitted}.

When $\os$ does not divide $\de$ then the ambient algebras are always not associative and all monomial isomorphisms between these algebras are monomial of degree one under the assumption that $\Aut(K)$ is commutative \cite[Corollary 6.2]{NevinsPumpluen2025submitted}.

In \cite{NevinsPumpluen2025submitted}  the number of Chen-equivalence and Chen-isometry classes of skew constacyclic codes of length $\de$ over $K=\mathbb{F}_{p^r}$ was computed.
(In \cite[Theorem 4]
{OuazzouNajmeddineAydin2025}, the  the number of Chen-isometry classes of codes of length $\de$ over $K=\mathbb{F}_{p^r}$ was computed incorrectly.)

%\textcolor{purple}{What was the definition of $[m]_s$ in the theorem below?  And is it $N$ that is the number of families?  Rather than $n$ as currently stated?}\textcolor{blue}{YES! Fixed it}

\begin{comment}
For an integer $s>0$, let
$$
[\de]_s=\frac{p^{s\de}-1}{p^s-1}=p^{s(\de-1)}+p^{s(\de-2)}+\cdots + p^s + 1.
$$

%Note that the proof does not require that $t^\de-a$ is reducible, so also counts the number of semifields which are not Chen-equivalent, or more precisely, gives an upper bound on their Chen equivalence classes, as not all of these algebras will be division algebras... This was our result:

\begin{theorem*}\label{c:Ouazzoufinite}
    Suppose $K=\mathbb{F}_{p^r}$ and $\sigma(x)=x^{p^s}$ with $s\mid r$ so that $\os=r/s$ and $\field_0=\mathbb{F}_{p^s}$.
    The number of distinct Chen isometry %$(\de,\sigma)$-isometry
    classes of families of skew constacyclic codes arising from nonassociative rings $\Kone{a}$ is $N$ where
    $$
     N= \begin{cases}
     \gcd([\de]_{s},p^r-1)&\text{if $\os\nmid \de$; and}\\
    \left(1-\frac{1}{[\os]_s}\right)\gcd([\de]_{s},p^r-1) & \text{if  $\os\mid \de$.}
    \end{cases}
    $$
    There are additionally $\gcd([\de]_{s},p^r-1)/[\os]_s$ different Chen equivalence % skew $(\de,\sigma)$-equivalence
    classes (and thus at most this many Chen-isometry classes) of families of skew $(\sigma,a)$-constacyclic codes arising from associative rings $\Kone{a}$, that is, for which $\os\mid \de$ and $a\in\field_0$.
\end{theorem*}
\end{comment}

By \cite[Corollary 6.2]{NevinsPumpluen2025submitted}, if   $a,b\in K^\times$ and $\os \nmid\de$, then any nonzero homomorphism
 $G:K[t;\sigma]/K[t;\sigma](t^\de-a) \to K[t;\sigma]/K[t;\sigma](t^\de-b)$  whose restriction to $K$ is given by some automorphism $\tau$ commuting with $\sigma$, must be monomial of degree one, hence in this case isometries are the same as equivalences, and indeed every isomorphism is an equivalence.

It is not always a simple matter to determine for which $a\in K^\times$ the skew polynomial $t^\de-a\in K[t;\sigma]$ is reducible, or in other words, for which $a\in K^\times$ nontrivial skew $(\sigma,a)$-constacyclic codes exist.  Applying Theorem~\ref{T:divalg}, we have the following result. % that precludes the existence nontrivial $(\sigma,a)$ skew constacyclic codes.  %When $\de$ is prime and $F$ contains a primitive $\de$th root of unity, \cite[Theorem 3.11]{BrownPhD2018} gives that $K[t;\sigma]/K[t;\sigma](t^\de-a) $ has zero divisors (that is, contains nontrivial codes)  if and only if there exists $\alpha\in K$ such that $N_\de^\sigma(\alpha)=a$.   Consequently we can  say more about the corresponding skew constacyclic codes.

\begin{proposition}\label{P:allcodesareconstacyclicifprime}
Let  $\de=2$ or $3$, or let $\de\geq 5$ be prime and suppose that $F$ contains a primitive $\de$th root of unity, that is $\mu_\de\subset F$. Then for any $a\in K^\times$, there exist nontrivial skew $(\sigma,a)$-constacyclic codes if and only if $\Kone{a}$ is equivalent to $\Kone{1}$, and in this case, any such code is isometric to a $\sigma$-constacyclic code.
%Then $K[t;\sigma]/K[t;\sigma](t^\de-a)$ has zero divisors if and only if it is Chen-equivalent  to $K[t;\sigma]/K[t;\sigma](t^\de-1)$.
\end{proposition}

\begin{proof}
  A nontrivial skew $(\sigma,a)$-constacyclic code is a  proper left ideal  in the Petit ring
$\Kone{a}$, generated by a monic right divisor of the skew polynomial $f(t)=t^\de-a$.  By Theorem~\ref{T:divalg}, %\cite[Theorem 3.11]{BrownPhD2018},
under the given  hypotheses, $K[t;\sigma]/K[t;\sigma](t^\de-a)$ has such ideals (that is, is not a division algebra)
 if and only if there exists $\alpha\in K$ such that
$N_\de^\sigma(\alpha)=a$.
%This in turn is equivalent to saying that $t-\alpha$ is  a right divisor of $t^\de-a$ for some $\alpha\in K$ by \cite{J96}.
This in turn is equivalent to  the existence of an isomorphism $$G_{\id,\alpha}:K[t;\sigma]/K[t;\sigma](t^\de-a)\to K[t;\sigma]/K[t;\sigma](t^\de-1).$$  As every such isomorphism induces an Hamming-weight-preserving bijection between the corresponding ideals in $K[t;\sigma]/K[t;\sigma](t^\de-a)$ and $K[t;\sigma]/K[t;\sigma](t^\de-1)$, it gives an isometry between the corresponding codes.  In particular, every nontrivial skew $(\sigma,a)$-constacyclic code is therefore isometric to a skew $(\sigma,1)$-constacyclic, that is, a skew cyclic, code.
  %We know that $\tau(a)=N_\de^\sigma(\alpha)$ is equivalent to $a=N_\de^\sigma(\tau^{-1}(\alpha))$, since $\tau^{-1}(N_\de^\sigma(\alpha))=N_\de^\sigma(\tau^{-1}(\alpha))$.
 %In other words, $K[t;\sigma]/K[t;\sigma](t^\de-a)$ is  isomorphic to $K[t;\sigma]/K[t;\sigma](t^\de-1)$ via  $G_{id,\alpha}$ if and only if $N_\de^\sigma(\alpha)=a$
%for some $\alpha\in K^\times$.
\end{proof}

Therefore, in the setting of Proposition~\ref{P:allcodesareconstacyclicifprime}, there is no advantage to using a skew $(\sigma,a)$-constacyclic code, as these are all isometric to the simpler skew cyclic codes.

That said, when $\de$ is not prime, there can be many skew $(\sigma,a)$-constacyclic codes where $K[t;\sigma]/K[t;\sigma](t^\de-a)\not\cong K[t;\sigma]/K[t;\sigma](t^\de-1)$, and classifying these equivalence and isometry classes of Petit algebras, that often contain many nontrivial codes, is the focus of our study.

\section{Counting the equivalence classes of skew constacyclic codes over finite fields} \label{sec:countingequivclasses}

Let $p$ be a prime, $\gamma\in \NN$ and $q=p^\gamma$.  Choose $\os,\de\in \mathbb{N}$. Set  $F_0=\mathbb{F}_p$, $F=\mathbb{F}_q$ and $K=\mathbb{F}_{q^\os}$. Then $\Gal(K/F)$ is cyclic of order $\os$, generated by $\sigma$ where $\sigma(x)=x^q$ for all $x\in K$.  Let $\xi$ denote a primitive element of $K$, meaning a generator of the cyclic multiplicative group $K^\times$ of $K$, which has order $q^\os-1$. We write $K^\times =\langle \xi \rangle$.

For any $x\in K^\times$ we define
$$
N^\sigma_\de(x) = \prod_{i=0}^{\de-1}\sigma^i(x)
$$
and set
$$
s = \gcd\left(q^\os-1, \frac{q^\de-1}{q-1}\right).
$$
Since as polynomials we have $\gcd(x^\os-1,x^\de-1)=x^{\gcd(\os,\de)}-1$, it follows that in particular,  $s>1$ whenever $\gcd(\os,\de)>1$.  Moreover, since $s\mid q^\os-1$ it cannot be divisible by $p$, we have that $\gcd(s,p^k)=1$ for all $k\in \NN$.% for all $\mathbb{F}_{\tilde{q}}$ of characteristic $p$.

\subsection{Chen equivalence classes}
Let $a,b\in K^\times$.
Recall from Definition~\ref{Definition:isometryChenisometryequivalenceChenequivalence}(iv) that we say that two algebras of the form $K[t;\sigma]/K[t;\sigma](t^\de-a)$ are Chen-equivalent if there exists an isomorphism of the form $G_{\id,\alpha}$ with $\alpha \in K^\times$, between them. %\cite{Pumpluen2025}.

Note that these Chen equivalence classes were counted, using different ideas, in \cite{NevinsPumpluen2025submitted}.

\begin{theorem}\label{L:Chenclassesfinitefield}
    The distinct Chen-equivalence classes of  ambient Petit rings of the form  $K[t;\sigma]/K[t;\sigma](t^\de-a)$, where $[K:F]=\os$, are parametrized by taking $a\in \{1, \xi, \cdots, \xi^{s-1}\}$, where $s = \gcd\left(q^\os-1, \frac{q^\de-1}{q-1}\right).$
    \end{theorem}

When  $[K:F]=\os$, the distinct Chen-equivalence classes ${\bf C}_a$ of skew constacyclic codes of length $n$ are thus parametrized by taking $a\in \{1, \xi, \cdots, \xi^{s-1}\}$, where $s = \gcd\left(q^\os-1, \frac{q^\de-1}{q-1}\right).$

\begin{proof}
    By \cite[Theorem 4.1]{Pumpluen2025}, an isomorphism of the form $G_{\id,\alpha}$ exists if and only if $a\in bN^\sigma_\de(K^\times)$. Therefore the Chen equivalence classes are parametrized by the left cosets of the subgroup $N_{\de}^\sigma(K^\times)$.  We compute
    $$
    N^\sigma_\de(\xi) = \xi \sigma(\xi) \cdots \sigma^{\de-1}(\xi) = \xi \xi^q \cdots \xi^{q^{\de-1}} = \xi^{\sum_{i=0}^{\de-1} q^i}=\xi^{(q^\de-1)/(q-1)}.
    $$
    Since $\xi$ has order $q^\os-1$, it follows that $N^\sigma_\de(K^\times)$ is the cyclic group generated by $\xi^s$, where $s=\gcd\left(q^\os-1, \frac{q^\de-1}{q-1}\right).$  Its cosets are thus parametrized by the set $\{1,\xi,\cdots, \xi^{s-1}\}$.
\end{proof}

In particular, if $\gcd(\os,\de)>1$ then $s>1$ and there is more than one Chen-equivalence class. For example, if $\os\mid \de$ then some of these Petit rings will be associative (those with $a\in F^\times$) and others not (those with $a\in K^\times\smallsetminus F^\times$), and these are not isomorphic so are never equivalent.

\begin{example}\label{Example:p=3n=2m=4}
Let  $p=q=3$, $\os=2$ and $\de=4$. Then $s=\gcd(8, 40)=8$ and  $N_\de^\sigma(\FF_9^\times)=\{1\}$.  Therefore  by Theorem~\ref{L:Chenclassesfinitefield},  no two algebras of the form $\FF_9[t;\sigma]/\FF_9[t;\sigma](t^4-a)$ are Chen-equivalent.
Those with $a\notin \FF_3$ are properly nonassociative.  In particular, note that in this case there are thus two associative Petit algebras, represented by $a=1$ and $a=2$, and these are non-Chen-equivalent.
\end{example}

At the other extreme we infer the following consequence.

\begin{corollary}\label{C:allskewequiv}
    Suppose $s = \gcd\left(q^\os-1, \frac{q^\de-1}{q-1}\right)=1$.  Then for every $a\in K^\times$, there exist nontrivial skew $(\sigma,a)$-constacyclic codes of length $\de$, and each one  is isometric to a skew cyclic code (that is, a skew $(\sigma,1)$-constacyclic code).
\end{corollary}

\begin{proof}
If $s=1$, then by Theorem~\ref{L:Chenclassesfinitefield},  all of the  ambient Petit rings $K[t;\sigma]/K[t;\sigma](t^\de-a)$, as $a$ ranges over $K^\times$, are Chen-equivalent to $$K[t;\sigma]/K[t;\sigma](t^\de-1).$$
In fact, the proof shows that in this case every $a\in K^\times$ is of the form $a=N^{\sigma}_\de(\alpha)$ for some $\alpha \in K^\times$, which implies that $(t-\alpha)$ right divides $t^\de-a$.  Consequently,  $t^\de-a$ is a reducible polynomial and well-suited for coding applications.  The isomorphism $G_{\id,\alpha}:K[t;\sigma]/K[t;\sigma](t^\de-a)\to K[t;\sigma]/K[t;\sigma](t^\de-1)$ maps each left principal ideal of the first ring (which corresponds to a choice of $(\sigma,a)$ skew constacyclic code) isometrically onto a left principal ideal of the second (which corresponds to a skew constacyclic code), whence the result.
\end{proof}

Note that in the setting of Corollary~
\ref{C:allskewequiv}, we must have $\gcd(\de,\os)=1$ and therefore it follows that  all the ambient Petit rings are properly nonassociative.

 %Thus  if $s=1$ then  $K[t;\sigma]/K[t;\sigma](t^\de-a)$ is not a division algebra for every $a\in K^\times$. %\textcolor{blue}{Moreover, if $\de$ is also prime and additionally there is an $\de$-th root of unity in $F$ then $t^\de-a$ is even decomposable into a product of $\de$ linear factors \cite[Theorem 3.10]{BrownPhD2018} or as Exercise in Bourbaki, Elements de Mathematique: Fasc. XXIII, Algebre, Chapitre 8, Modules et anneaux semi-simples p.~344.} \textcolor{brown}{How did we know that $t-\alpha$ right divides $t^\de-a$? I naively thought we'd get $t^\de-a = \prod_i (t-\sigma^i(a))$ but now realize I don't know how to show that... I should look at your Bourbaki... }

%In general we have:  an isomorphism of the form $G_{id,\alpha}$ exists if and only if  $ab^{-1}= N^\sigma_\de(\alpha)$ if and only if $t-\alpha$ right divides $t^\de-ab^{-1}$.

%This reassures us that under some circumstances, the only skew constacyclic codes of interest are equivalent to $a=1$.

Note that this is in some sense a strengthening of Proposition~\ref{P:allcodesareconstacyclicifprime}: when $s=1$, it is not necessarily to consider the primality of $\de$ in order to decide if skew $(\sigma,a)$-constacyclic codes exist: they do, for each $a\in K^\times$.

\begin{corollary}
    Suppose $s=\gcd\left(q^\os-1, \frac{q^\de-1}{q-1}\right)=1$. Then $\Kone{a}$ is not a division algebra, that is, $t^\de-a$ is reducible.
\end{corollary}

\begin{proof}
    Under the given hypothesis, $\Kone{a}\cong \Kone{1}$, which has nontrivial principal left ideals.
\end{proof}

\subsection{Equivalence classes}

We next turn to the question of counting classes under the stronger notion of \emph{equivalence} in the sense of Definition~\ref{Definition:isometryChenisometryequivalenceChenequivalence} (iii), that is, with respect to isomorphisms of the form $G_{\tau,\alpha}$ where $\tau\in \Aut(K)=\Gal(K/F_0)$ and $\alpha\in K^\times$.  When the fixed field of this automorphism $\tau$ contains $\tilde{F}$, then we say that the rings are \emph{equivalent over $\tilde{F}$}.  We retain this slight generality in this section for interest, but the key case the reader should bear in mind is $\tilde{q}=p$, that is,  when $\tilde{F}=\FF_p$.  Algebras that are equivalent over $\FF_p$ are in fact equivalent as rings.

For each $d\in \NN$, write $\varphi(d)=|(\Z/d\Z)^\times|$ for the Euler phi function of $d$. For any $\ell$ such that $\gcd(\ell,d)=1$, write $o(\ell)_d$ for the order of the image of $\ell$ in $(\Z/d\Z)^\times$.

\begin{proposition}\label{T:F-tildeequivalencecount}
    Let $\tilde{F}$ be a subfield of $F$ and denote its size by $\tilde{q}$.  The number of equivalence classes over $\tilde{F}$ of algebras of the form $K[t;\sigma]/K[t;\sigma](t^\de-a)$, as $a$ ranges over $K^\times$ is
    $$
    \sum_{d|s}\frac{\varphi(d)}{o(\tilde{q})_d}.
    $$
\end{proposition}

\begin{proof}
    Let $N=[K:\tilde{F}]=\log_p(q^\os)/\log_p(\tilde{q})$ denote the degree of this extension.  The automorphism of $K$ defined by $\tau(x)=x^{\tilde{q}}$ for each $x\in K$ generates the cyclic  Galois group ${\rm Gal}(K/\tilde{F})$ and commutes with $\sigma$.  In particular, $\tau$ stabilizes $K^\times$ and the subgroup $N^\sigma_\de(K^\times)$, so acts as a permutation on  the cosets of $N^\sigma_\de(K^\times)=\langle \xi^s\rangle$, where $s=\gcd\left(q^\os-1, \frac{q^\de-1}{q-1}\right)$.
    By Lemma~\ref{L:Chenclassesfinitefield}, each coset is represented by $\xi^i$ for some $0\leq i <s$.  For each such $i$, the orbit of $\xi^i$ under the Galois group is
    $$
    \Gal(K/\tilde{F})\cdot \xi^i = \{ \xi^{i\tilde{q}^r} \mid 0\leq r < N\}.
    $$
    Thus the Galois orbit of the coset $\xi^i N^\sigma_\de(K^\times) = \xi^i\langle \xi^s\rangle$ in $K^\times$ is the union
    $$
    \Xi_i:=\bigcup_{0\leq r<N}\xi^{i\tilde{q}^r}N^\sigma_\de(K^\times).
    $$
    These orbits need not be distinct: for example if $j\equiv i\tilde{q}\mod s$ then $\xi^j$ and $\xi^{i\tilde{q}}$ represent the same coset, so $\Xi_i=\Xi_{j}$.  The orbits are also not all of the same size: for example, the trivial coset forms its own orbit, whereas in general the size depends on the order of $\tilde{q}$ mod $s$.  Nevertheless, each of the \emph{distinct} orbits  $\Xi_i$ defines a single equivalence class of algebras $K[t;\sigma]/K[t;\sigma](t^\de-a)$ over $\tilde{F}$, namely, by letting $a$ run over $\Xi_i$.

    To count the number of distinct orbits, we proceed as follows.  Let $d$ be a divisor of $s$.  Set ${\mathbf R}_d = \{i\in \Z/s\Z\mid \gcd(i,s)=s/d\}$ to be the set of all elements of additive order exactly $d$; since $\gcd(\tilde{q},s)=1$, we have that for each $i\in {\mathbf R}_d$ and $0\leq r<N$ that $\gcd(i\tilde{q}^r,s)=s/d$ so
    $$
    \Xi(d) = \bigcup_{j\in \mathbf{R}_d} \xi^j N^\sigma_\de(K^\times)
    $$
    is the $\Gal(K/\tilde{F})$-invariant union of the orbits $\Xi_i$ for which $\gcd(i,s)=s/d$.  Moreover, since the sets $\mathbf{R}_d$ partition $
    \Z/s\Z$, we have
    $$
    K^\times = \bigsqcup_{d|s} \Xi(d)
    $$
    so it suffices to count the number of distinct $\Gal(K/\tilde{F})$ orbits in each such $\Xi(d)$.

    Fix $d$ dividing $s$.  Since  $\gcd(d,\tilde{q})=1$, we may set $s_0:=o(\tilde{q})_d$ to be the order of $\tilde{q}$ in $(\Z/d\Z)^\times$. Thus for each $i\in {\mathbf R}_d$,  $s_0$ is the least power of $\tilde{q}$ for which we have $i(\tilde{q}^{s_0}-1)\equiv 0\mod s$, whence it is the least power such that $i\tilde{q}^{s_0}\equiv i\mod s$. Therefore for each $i\in {\mathbf R}_d$, the number of distinct cosets in the orbit $\Xi_i$ is $s_0$.  Since the map $k\mapsto k(s/d)$ is an isomorphism of $(\Z/d\Z)^\times$ onto ${\mathbf R}_d$, $\Xi(d)$ contains precisely $\varphi(d)$ distinct cosets, which are partitioned into $\varphi(d)/s_0$ orbits of size $s_0=o(\tilde{q})_d$, as required.
\end{proof}

\begin{example}
    In the setting of Example~\ref{Example:p=3n=2m=4}, we have $s=8$, that is, there are exactly $8$ Chen equivalence classes. Setting $\tilde{F}=F_0=\FF_3$ in Proposition~\ref{T:F-tildeequivalencecount} we deduce that there are only
    $$
    \sum_{d|8}\frac{\varphi(d)}{o(3)_d} = \underbrace{\frac{1}{1}}_{d=1}+\underbrace{\frac{1}{1}}_{d=2}+\underbrace{\frac{2}{2}}_{d=4}+
    \underbrace{\frac{4}{2}}_{d=8}=5
    $$
    equivalence classes of these algebras.  To describe these equivalence classes, set $${\mathbf R}_d = \{i\in \Z/s\Z\mid \gcd(i,s)=s/d\}$$ as in the proof.  We compute that
    $$
    {\mathbf R}_1=\{0\}, {\mathbf R}_2=\{4\}, {\mathbf R}_4=\{2,6\},\;\text{and}\; {\mathbf R}_8=\{1,3,5,7\}.
    $$
    Now the equivalence classes are given by the orbits of the Galois group on the sets
    $$
    \Xi(d) = \bigcup_{j\in \mathbf{R}_d} \xi^j N^\sigma_\de(K^\times),
    $$
    which we compute as follows.

    Since $N_\de^\sigma(\FF_9^\times)=\{1\}$, each coset $\xi^jN_\de^\sigma(\FF_9^\times)$ is simply the set $\{\xi^j\}$.  Thus $\Xi(1) = \{\xi^0=1\}$ is one equivalence class, and $\Xi(2)=\{\xi^4\}$ is another.  We have $\Xi(4)=\{\xi^2,\xi^6=\sigma(\xi^2)\}$ which forms a single orbit under the Galois group, hence another equivalence class.  Finally, $\Xi(8)=\{\xi,\xi^3=\sigma(\xi),\xi^5,\xi^7=\sigma(\xi^5)\}$, which we see forms two orbits under the Galois group, hence two equivalence classes, as predicted by the theorem.

    Note that the two associative algebras are represented by $a=1$ and $a=2=\xi^4$, and these are inequivalent.
\end{example}

In general, computing the order of $\tilde{q}$ modulo various divisors $d$ of $s$ is an arithmetically interesting problem.  It is simple in some cases; for example, for all $d|\gcd(s,\tilde{q}-1)$ we have $\tilde{q}\equiv 1\mod d$ and so $o(\tilde{q})_d=1$.

While Proposition~\ref{T:F-tildeequivalencecount} is stated for a general base field $\tilde{F}$, the core case is when we choose $\tilde{F}=F_0=\FF_p$ and $\tilde{q}=p$;  this is exactly the case we need to consider in order to count all the equivalence classes of skew constacyclic codes.

\begin{theorem}\label{them:countingskewconstacodes}
The number of equivalence classes of ambient Petit rings  $\Kone{a}$, with $a\in K^\times$ and
with $\sigma$ of order $\os$, that is, the number of  equivalence classes  ${\bf C}_a$, is
    $$
    \sum_{d|s}\frac{\varphi(d)}{o(p)_d}.
    $$
In particular, when  $\os \nmid\de$ we conclude that this is also the number of \emph{isometry} classes  of these rings, i.e. the number of  \emph{isometry} classes  ${\bf C}_a$.
\end{theorem}

\begin{proof}
This is just the specialization to $\tilde{F}=\FF_p$ in Proposition~\ref{T:F-tildeequivalencecount}. For the final statement, note that when   $\os \nmid\de$, isometry and equivalence coincide, by Theorem~\ref{T:isometryandequivalencecoincide}.
\end{proof}

The value in Theorem~\ref{them:countingskewconstacodes} can be arithmetically difficult to compute when $s$ is very large, as it is a sum of positive integer terms over all the divisors of $s$.  We have, however, the following immediate corollary.

\begin{corollary}\label{cor:lowerbound}
    The number of equivalence classes of ambient Petit rings of the form $\Kone{a}$, with $a\in K^\times$ and
with $\sigma$ of order $\os$, is bounded below by $\mathsf{d}(s)$, where $\mathsf{d}$ is the number-of-divisors function \cite[https://oeis.org/A000005]{OEIS}.
\end{corollary}

\begin{proof}
    Each term of the sum in Theorem~\ref{them:countingskewconstacodes} is a positive integer, so greater than or equal to $1$.  Hence the sum is at least equal to the number of divisors of $s$.
\end{proof}

To state another useful corollary, we need the following arithmetic lemma.

\begin{lemma}\label{L:s_0}
Suppose $\os\mid\de$.    Let $s=\gcd\left( q^\os-1, \frac{q^{\de}-1}{q-1}\right)$ and $s_0 = \gcd(q-1,\frac{\de}{\os})$. Then
    $$
    s = \left(\frac{q^\os-1}{q-1}\right) s_0.
    $$
\end{lemma}
\begin{proof}
Since $\os\mid\de$, we have  $q^{\os}-1$ divides $q^{\de}-1$.  Note that
$$
s = \gcd\left( q^\os-1, \frac{q^{\de}-1}{q-1}\right) = \left( \frac{q^\os-1}{q-1}\right)\gcd\left(q-1,\frac{q^{\de}-1}{q^{\os}-1}\right).
$$
Since $(q^\de-1)/(q^\os-1)\equiv \de/\os\mod q-1$, we have
$$
s_0 = \gcd\left(q-1,\frac{q^{\de}-1}{q^{\os}-1}\right) =\gcd\left(q-1,\frac{\de}{\os}\right).
$$
\end{proof}

We can now count the equivalence classes of \emph{associative} Petit algebras $\Kone{a}$, that is, such that $\os\mid\de$ and $a\in F^\times$.

\begin{corollary}\label{Cor:associativecase}
    Suppose $\os\mid \de$ and set $s_0 = \gcd(q-1,\frac{\de}{\os})$. Then the number of equivalence classes of ambient Petit rings $\Kone{a}$ that are \emph{associative}, meaning, as $a$ runs over $F^\times$, is given by
    $$
    \sum_{d|s_0}\frac{\varphi(d)}{o(p)_d}.
    $$
\end{corollary}

\begin{proof}
    When $\os\mid\de$, the associative algebras are those for which $a\in F^\times$.  Since
    $$N_\de^\sigma(K^\times) = N_{K/F}(K^\times)^{\de/\os}=(F^\times)^{\de/\os}=(F^\times)^{s_0},
    $$
    and $\Gal(K/\FF_p)$ preserves the cosets space $F^\times/(F^\times)^{s_0}$, we may imitate the argument of the proof of Proposition~\ref{T:F-tildeequivalencecount}.  In particular, since $F^\times = \langle \xi^{\left(\frac{q^\os-1}{q-1}\right)}\rangle$, we have that $F^\times$ is the union of the sets $\mathbf{R}_d$ corresponding to the divisors $d$ of $s_0$, since
    $$R^0_d=\{j\in \ZZ/s_0\ZZ \mid \gcd(j,s_0)=s_0/d\}= \{i=\left(\frac{q^\os-1}{q-1}\right)j\mid \gcd(i,s)=s/d\}=\mathbf{R}_d.$$
    The number of $\Gal(K/\FF_p)$-orbits in these sets was computed in Proposition~\ref{T:F-tildeequivalencecount}, yielding the result.
\end{proof}

We will refer to this corollary in Section~\ref{sec:examples}, when we show that isometry and equivalence do not always coincide.

\subsection{An algorithm to parametrize
the equivalence classes}

Let $\tilde{F}\subset F\subset K$.
If $a,b\in K^\times$ satisfy that $\Kone{a}$ is equivalent to $\Kone{b}$ (via some isomorphism $G_{\tau,\alpha}$) we may write $a\cong_{equiv} b$.  Therefore, our goal is to generate
a set of elements $a\in K^\times$ for the equivalence classes under $\cong_{equiv}$, that is, that correspond to exactly one representative of each equivalence class of Petit algebras $K[t;\sigma]/K[t;\sigma](t^\de-a)$.  This is given in Algorithm~\ref{Algorithm:1}.

We explain it briefly as follows.  Given $\os,\de$, $K=\mathbb{F}_{q^\os}$ and $F=\mathbb{F}_q$, we compute $s = \gcd(q^\os-1, \frac{q^{\de}-1}{q-1})$ and let $\xi$ be a generator of $K^\times$.  If $s=1$ then all ambient algebras $K[t;\sigma]/K[t;\sigma](t^\de - a)$ are Chen-equivalent (hence equivalent) and  all $t^\de-a$ are reducible. In this case we have exactly one equivalence class of  ambient algebras  for skew constacyclic codes, represented by $K[t;\sigma]/K[t;\sigma](t^\de - 1)$ and we are done.

Otherwise, we cycle through the $\Gal(K/\tilde{F})$-orbit of each element $\xi^i$, as $i$ runs from $0$ to $s-1$ (Lemma~\ref{L:Chenclassesfinitefield}), and identify the resulting coset of $N_\de^\sigma(K^\times)$.  If it is not equal to the coset of $\xi^i$, then that coset is removed from the list.  In this way, we retain a unique representative of each of these orbits.

Consequently, Algorithm~\ref{Algorithm:1} generates representatives for the distinct equivalence classes over a chosen subfield $\tilde{F}$ of $F$, thus in particular representatives for the distinct equivalence classes of ambient rings of skew $(\sigma,a)$ constacyclic codes when $\tilde{F}=\mathbb{F}_p$.

\begin{algorithm}
\caption{Generating representatives of the distinct %$1$-
equivalence classes over $\tilde{F}$\label{Algorithm:1}}
\begin{algorithmic}[1]
\Input : $K=\FF_{q^\os}$, $F=\FF_q$, $\tilde{F} = \mathbb{F}_{\tilde{q}}\subset F\subset K$ and a length $\de$
\EndInput
\Output : A set $P \subseteq \{0,1,\ldots,s-1\}$ such that the distinct equivalence classes of Petit algebras over $\tilde{F}$ are precisely $K[t;\sigma]/K[t;\sigma](t^\de - \xi^r)$ as $r$ runs over $P$
\EndOutput
\State $s\gets \gcd(q^\os-1, \frac{q^{\de}-1}{q-1})$
\If{$s=1$}
\State \Return $P=\{0\}$ \Comment{All are (Chen-)equivalent to $\xi^0=1$}
\Else
\State $P \gets \{0\}$ \Comment{Initialization of the set of representatives}
\State $Q \gets \emptyset$ \Comment{Values $1\leq i<s$ for which $\xi^i \in \xi^rN^\sigma_\de(K^\times)$ for some $r\in P$}
\State $k \gets 1$ \Comment{Start with $\xi$, which satisfies $\xi\not\cong_{equiv} 1$}
\While{$k < s$}
    \If{$k \notin Q$} \Comment{If $k$ is not a previous orbit:}
        \State $P \gets P \cup \{k\}$ \Comment{Record $k$ as a distinct orbit}
        \State $j \gets k\tilde{q} \bmod s$ \Comment{Go through its $\Gal(K/\tilde{F})$-orbit}
        \While{$j \notin Q$ \textbf{and} $j \notin P$}
            \State $Q \gets Q \cup \{j\}$
            \State $j \gets j\tilde{q} \bmod s$
        \EndWhile \Comment{Until we loop back to a previously-sorted value}
    \Else
        \State $k \gets k + 1$  \Comment{Advance to next non-sorted value}
    \EndIf
\EndWhile
\State \Return $P$
\EndIf
\end{algorithmic}
\end{algorithm}

\section{Parametrizing isometry classes of Petit rings}\label{sec:isoclasses}

The next question is to further refine this parametrization of Petit rings to one of isometry classes, where isometry is defined as \emph{any} Hamming-weight preserving isomorphism in Definition~\ref{Definition:isometryChenisometryequivalenceChenequivalence}(i).   We begin by developing some theory.

\subsection{Results relating isometry and equivalence}\label{SS:theoremsforisometry}
To begin, note that by Theorem~\ref{T:isometryandequivalencecoincide}, isometry and equivalence coincide for all properly nonassociative Petit algebras, that is, whenever $\os\nmid \de$ and $a \notin F^\times$.  To characterize isometry in the remaining cases, we further specialize a result from \cite{NevinsPumpluen2025submitted} to the case of finite fields to obtain the following more precise criterion.  We write $(F^\times)^z = \{c^z\mid c\in F^\times\}$ for the subgroup of $z$th powers in $F^\times$.

\begin{theorem}\label{T:isometrydifferentfromequivalence}
Let $K=\FF_{q^\os}$ and $F=\FF_q$ be finite fields of characteristic $p$.
Suppose $\os\mid\de$ and set $s_0=\gcd(\de/\os,q-1)$.  Then two \emph{inequivalent} Petit rings $K[t;\sigma]/K[t;\sigma](t^\de-a)$ and $K[t;\sigma]/K[t;\sigma](t^\de-b)$ are isometric over $\FF_p$  if and only if
     $a,b\in F^\times$ and
there exists $\tau \in \Gal(K/\FF_p)$ such that $$\tau(a)\in b^k (F^\times)^{s_0}$$
for some $k\in \NN$ satisfying
\begin{equation}\label{E:conditionsonk}
2\leq k <\de, \quad k\equiv 1 \mod \os, \quad \text{and} \quad \gcd(k,\de)=1.
\end{equation}
\end{theorem}

\begin{proof}
Two Petit rings are isometric but not equivalent exactly when there exists a Hamming-weight preserving isomorphism of degree strictly greater than one (and no isomorphism of degree equal to one) between them.
    The result \cite[Theorem 4.4]{NevinsPumpluen2025} characterizes the Petit rings over an arbitrary commutative unital ring $S$ such that such an isometry $G_{\tau,\alpha,k}:\Sone{a}\to\Sone{b}$ exists, namely, one must satisfy all of the following:
    \begin{enumerate}[(i)]
    \item $\tau\in \Aut(\ring)$ commutes with $\sigma$;
    \item $\alpha \in \ring^\times$
    \item $\os\mid \de$;
        \item  $k\equiv 1 \mod \os$;
        \item  $\gcd(k,\de)=1$;
                \item $a, b\in \ring_0$;
        \item  $(N_{\ring/\ring_0}(\alpha))^{\de/\os}b^k=\tau(a)$
    \end{enumerate}
In our case, $S=K$ and $S_0=F$.  Then we must satisfy the four conditions (iii) $\os|\de$, (iv) $k\equiv 1 \mod \os$, (v) $\gcd(k,\de)=1$ and (vi) $a,b\in F^\times$, which implies in particular that $\Kone{a}$ is associative.  Since $\Aut(K)=\Gal(K/\FF_p)$ is abelian, (i) always holds.  Since the norm map on finite fields is surjective, $\{N_{\ring/\ring_0}(\alpha))^{\de/\os}\mid \alpha\in \ring^\times\}= (F^\times)^{\de/\os}$, the group of $(\de/\os)$ powers.  Since $F^\times$ is cyclic of order $q-1$, we have $(F^\times)^{\de/\os}=(F^\times)^{s_0}$ where $s_0=\gcd(\de/\os,q-1)$.  Therefore the remaining conditions (ii) and (vii) are satisfied if and only if $\tau(a)\in b^k(F^\times)^{s_0}$.
\end{proof}

Let us now apply this result to derive an algorithm to construct a set of representatives for the isometry classes of Petit algebras.

By Theorem~\ref{T:isometrydifferentfromequivalence}, Algorithm~\ref{Algorithm:1} may overcount isometry classes only if $\os\mid\de$, and then the only parameters $a$ in question are those lying in $F^\times$.

Suppose $\os\mid\de$.
If $K^\times = \langle \xi \rangle$, then $F^\times = \langle \xi^{(q^\os-1)/(q-1)}\rangle$.  Therefore  in the language of Algorithm~\ref{Algorithm:1}, we need to determine which elements of the set
$$
P \cap \frac{q^\os-1}{q-1}\ZZ
$$
correspond to isometric (but inequivalent) Petit algebras.  Here, $P\subset \{0,1,\cdots, s-1\}$ is the set of powers $i$ such that $\xi^i$ runs over the equivalence classes (Galois orbits of cosets of $N_\de^\sigma(K^\times)$), and $\frac{q^\os-1}{q-1}\ZZ$ is the set of powers corresponding to elements of $F^\times$.

We record our first observation as the following lemma.

\begin{lemma}\label{le:s}
If $\os\mid\de$ and $s_0=\gcd(q-1,\de/\os)=1$, %$s=\frac{q^\os-1}{q-1}$ (or equivalently, )
then equivalence and isometry coincide for all ambient Petit algebras of skew constacyclic codes.
\end{lemma}

\begin{proof}
    Applying Lemma~\ref{L:s_0}, note that the intersection $P \cap \frac{q^\os-1}{q-1}\ZZ$ is a singleton whenever $s_0=1$, since in this case, $s=\frac{q^\os-1}{q-1}$ and $P\subset \{0,1,\cdots,s-1\}$.  In particular, we deduce that isometry classes coincide with equivalence classes for all associative Petit algebras in this case (that is, those with $a\in F^\times$). Together with Theorem~\ref{T:isometryandequivalencecoincide}, which gives that  isometry and equivalence coincide for all properly nonassociative Petit algebras, we have the result.
\end{proof}

Next, note that the conditions on the degree of the isometry $k$ imposed by Theorem~\ref{T:isometrydifferentfromequivalence} are stringent, and the set
\begin{equation}\label{E:S}
S = \{ 1+r\os \mid 1\leq r < \de/\os, \gcd(1+r\os,\de)=1\}
\end{equation}
of elements $k$ satisfying the conditions \eqref{E:conditionsonk}
can be empty.

\begin{example}
    Suppose $\os=3$ and $\de=6$. Then $\{1+r\os\mid 1\leq r < \de/\os=2\} = \{4\}$ but $\gcd(4,\de)=2$, so $S=\emptyset$.  Therefore, the only possible isometries between Petit rings of the form $\Kone{a}$ have degree $1$, that is, are given by equivalence.
\end{example}

Furthermore, some values of $k$ define isometries that can already be realized by an isomorphism that is an equivalence.

\begin{lemma}\label{L:nonewisometries}
    Let $q=p^\gamma$ and $s_0=\gcd(q-1,\de/\os)$. Let $k\in S$, where $S$ was defined in \eqref{E:S}.  Suppose there exists $\ell \in [0,\gamma-1]$ such that $k\equiv p^\ell\mod s_0$.  Then if there exists an isometry of degree $k$ between two ambient Petit algebras $\Kone{a}$ and $\Kone{b}$, for some $a,b\in F^\times$, then $a\cong_{equiv}b$, that is, these algebras are already equivalent (that is, isometric via an isometry of degree $1$).
\end{lemma}

\begin{proof}
    Suppose we are given such a $k$.  By Theorem~\ref{T:isometrydifferentfromequivalence}, it defines, for each $\alpha \in K^\times$ and $\tau\in \Gal(K/\FF_p)$, the  isometries $G_{\tau,\alpha,k}:\Kone{a}\to \Kone{b}$ where $a$ and $b$ are related by the condition $N_{K/F}(\alpha)^{\de/\os}b^k=\tau(a)$.
    We claim that in fact $a\cong_{equiv}b$.

    Namely, let $\tilde{\tau}\in \Gal(F/\FF_p)$ be given by $\tilde{\tau}(b)=b^{p^\ell}$ for all $b\in F^\times$.  Since $k\equiv p^\ell \mod s_0$, there exists $t\in \ZZ$ such that $p^\ell=k + ts_0$. Recalling that $N_{K/F}(K^\times)^{\de/\os}=\langle \xi^s\rangle = (F^\times)^{s_0}$, we deduce that
    $$b^{p^\ell}=b^k  (b^{s_0})^t$$
    so that $$
    \tilde{\tau}(b)\in b^kN_{K/F}(K^\times)^{\de/\os}.$$
    Therefore, we have that $b^k\cong_{equiv} b$.  On the other hand, $N_{K/F}(\alpha)^{\de/\os}b^k=\tau(a)$ implies directly that $b^k \cong_{equiv} a$.  By transitivity, $a\cong_{equiv}b$, that is, the two Petit algebras were already isomorphic via a simple equivalence.
\end{proof}

\begin{example}
    Suppose $p=2$, $q=4$, $\de=12$ and $\os=4$.  Then $s_0=\gcd(q-1,\de/\os)=3$.  Here, $p\equiv 2\mod s_0$ and $p^2\equiv 1\mod s_0$.  We find $S=\{5\}$ has a single element, but since $k=5\equiv p\mod s_0$, Lemma~\ref{L:nonewisometries} implies that there will not exist any elements $a,b$ giving isometric Petit algebras such that $a\not\cong_{equiv}b$.  Hence, isometry classes and equivalence classes coincide.
\end{example}

\subsection{An algorithm for parametrizing isometry classes}
Summarizing the results from Section~\ref{SS:theoremsforisometry}, we see that it suffices, in the case that $s_0>1$, $S\neq \emptyset$ and $S$ containing elements $k$ that do not satisfy Lemma~\ref{L:nonewisometries}, to determine which of the elements of $P \cap \frac{q^\os-1}{q-1}\ZZ$ yield isometric Petit algebras, and find a set of representatives $P_0$ of the resulting classes.  This is done by removing the elements of $P_0$ that arise as $a^k$ for some $a\in P_0$ and $k\in S$, as in this case, $a$ and $a^k$ yield isometric algebras.

Then a complete set of isometry classes of Petit algebras is parameterized by
\begin{equation}\label{E:replaceP_0}
P_{iso}:=(P \smallsetminus \frac{q^\os-1}{q-1}\ZZ) \cup P_0,
\end{equation}
where $P \smallsetminus \frac{q^\os-1}{q-1}\ZZ$ parametrizes the properly nonassociative Petit algebras up to isometry and $P_0$ parametrizes the associative ones up to isometry.

This process is laid out in Algorithm~\ref{Algorithm:2}.

\begin{algorithm}
\caption{Computation of representatives of the distinct isometry classes over $\FF_p$\label{Algorithm:2}}
\begin{algorithmic}[1]
\Input : $K=\FF_{q^\os}\supset F=\FF_q=\FF_{p^\gamma} \supseteq \FF_p$; $\tau(x)=x^p$ a generator of $\Gal(F/\FF_p)$, of order $\gamma$; a code length $\de$ such that $\os\mid\de$ % and $s=\gcd(q^\os-1, \frac{q^{\de}-1}{q-1})$
\EndInput
\Output : A set $P_{iso} \subseteq \{0,1,\ldots,s-1\}$ such that the distinct isometry classes of Petit algebras over $\FF_p$ are precisely $K[t;\sigma]/K[t;\sigma](t^\de - \xi^r)$ as $r$ runs over $P_{iso}$
\EndOutput
\State $P\gets$ Algorithm~\ref{Algorithm:1} \Comment{Get representatives for all equivalence classes}
\State $s_0 \gets \gcd\left(q-1,\de/\os\right)$  \Comment{Number of equivalence classes from $F$}
\State $S\gets \emptyset$
\For{$1\leq r<\de/\os$} \Comment{Determine the set $S$ of possible degrees $k$ of the isometries}
\If{$\gcd(1+r\os,\de)=1$ {\bf and} $\forall\ell\in[0,\gamma], 1+r\os\neq p^\ell \mod s_0$}
\State $S\gets S\cup \{1+r\os\}$
\EndIf
\EndFor
\If{$\os\nmid\de$ {\bf or} $s_0=1$ {\bf or} $S=\emptyset$}
\State \Return $P_{iso}=P$ \Comment{Isometry and equivalence coincide}
\Else
\State $P_0 \gets P\cap \left\{ 0,\frac{q^{\os}-1}{q-1}, 2(\frac{q^{\os}-1}{q-1}), \cdots, (q-2)\frac{q^{\os}-1}{q-1}\right\}$  \Comment{subset to modify}
\State $i\gets 0$
\While{$i<\text{length}(P_0)$}
    \State $\alpha \gets (P_0)_i$ \Comment{Take the $i$th element of $P_0$}
    \For{$k\in S$ {\bf and} $\ell\in [0,\gamma-1]\in \ZZ$}
        \State $\alpha'\gets \alpha k p^\ell \mod s$ \Comment{Coset representative of $\tau^\ell(a^k)$}
        \If {$\alpha\neq \alpha'$ and $\alpha'\in P_0$} \Comment{if $\alpha'$ is a distinct equivalence class}
            \State $P_0 \gets P_0\setminus \{\alpha'\}$
        \EndIf
    \EndFor
    \State $i\gets i+1$ \Comment{move to next element of $P_0$}
\EndWhile
\State $P_{iso}\gets (P\smallsetminus \left\{ 0,\frac{q^{\os}-1}{q-1}, 2(\frac{q^{\os}-1}{q-1}), \cdots, (q-2)\frac{q^{\os}-1}{q-1}\right\}) \cup P_0$ \Comment{As in \eqref{E:replaceP_0}}
\State \Return $P_{iso}$
\EndIf
\end{algorithmic}
\end{algorithm}

\section{Constructing examples of the difference between isometry and equivalence}\label{sec:examples}

We begin with three observations to refine the set of examples we seek in this section.
\begin{enumerate}[1.]
\item Recall that the skew constacyclic codes of length $\de$  over $\mathbb{F}_{q^\os}$ will be left  ideals of  Petit rings of the form $\Kone{a}$; there are nontrivial $(\sigma,a)$ skew constacyclic codes exactly when  $t^\de-a$ is reducible.

\item When $\Kone{a}$ is isometric with $\Kone{1}$ (see Proposition~\ref{P:allcodesareconstacyclicifprime} and Corollary~\ref{C:allskewequiv} for important cases where this holds very broadly), then
all skew $(\sigma,a)$ constacyclic codes are equivalent to $(\sigma,1)$ constacyclic codes. Thus the interesting new examples are to be found among the remaining isometry classes.

\item Finally, although under various hypotheses (see Theorem~\ref{T:isometryandequivalencecoincide}, Lemma~\ref{le:s} and Lemma~\ref{L:nonewisometries})  we have that isometry and equivalence coincide, this is not the case in general.
\end{enumerate}

Therefore our interest is in rings $\Kone{a}$ such that $t^\de-a$ is reducible, $\Kone{a}$ is not isometric with the ring $\Kone{1}$, and such that there exists $b$ such that $\Kone{a}\not\cong_{equiv}\Kone{b}$ but $\Kone{a}$ and $\Kone{b}$ are indeed isometric.

\subsection{Conditions on \texorpdfstring{$\de,\os,q,p$}{m,n,q,p}}
To generate such examples, let us explore the conditions of Lemma~\ref{L:nonewisometries} more carefully.  The values $k$ that contribute no new isometries are those for which the image of $k\mod s_0$ lies in the set $\mathscr{P}:=\{1,p,p^2,\cdots,p^{\gamma-1}\}\mod s_0$, where   $q=p^\gamma$.  This leads us to the following complement to Theorem~\ref{T:isometryandequivalencecoincide}.

\begin{proposition}\label{prop:proper equivalences}
    Suppose $\os\mid\de$ and set $S = \{ 1+r\os \mid 1\leq r < \de/\os, \gcd(1+r\os,\de)=1\}$.  There exists $a\in F^\times$ such that $\Kone{a}$ and $\Kone{a^k}$ are isometry but not equivalent if
    all the following hold:
    \begin{itemize}
        \item $s_0=\gcd(\de/\os,q-1)>1$;
        \item the subgroup $\mathscr{P}$ of $(\ZZ/s_0\ZZ)^\times$ generated by $p$ is proper (that is, $p$ is not a primitive element);
        \item there exists $k\in S$ %$k\in [2,m)$ such that $k\equiv 1\mod n$, $\gcd(k,m)=1$ and
        such that the image of $k$ in $(\ZZ/s_0\ZZ)^\times$ is not in $\mathscr{P}$.
    \end{itemize}
    Note that the last condition implies $S$ is nonempty and that $s_0\nmid n$.
\end{proposition}

\begin{proof}
We first show that each of these conditions are necessary.
    If $s_0=1$ then $N_\de^\sigma(K^\times)=F^\times$ and thus all associative algebras are Chen-equivalent.  If $\mathscr{P}=(\ZZ/s_0\ZZ)^\times$ then the Frobenius map will act transitively on the representatives of the nontrivial cosets of $F^\times/(F^\times)^{s_0}$, whence all associative algebras are equivalent.
The final condition specifies the existence of an integer $k$ such that $G_{\id,1,k}$ defines an isomorphism.  Since $(F^\times)^{s_0k}\subseteq (F^\times)^{s_0}$, the map $a\mapsto a^k$ permutes the coset space $F^\times/(F^\times)^{s_0}$.  However, if all such $k$ lie in $\mathscr{P}$ then this map $a\mapsto a^k$ will always be realized by some element $\tau$ of $\Gal(F/\mathbb{F}_p)$, meaning the resulting algebras are already equivalent, that is, isomorphic via  $G_{\tau,\alpha}$ for some $\alpha\in K^\times$.

If all conditions hold, then there exists an element $a\in F^\times$ and $k\neq 1$ such that $a^k\notin \tau(a)N_\de^\sigma(K^\times)$ for any element $\tau\in\Aut(K)$, whence $\Kone{a^k}$ is not equivalent to $\Kone{a}$ and yet $G_{\id,1,k}$ defines an isometry between these algebras.

The final statement includes the observation that if $s_0\mid \os$, then any $k\in S$ satisfies $k\equiv 1 \equiv p^0\mod s_0$, so that the third condition fails.
\end{proof}

\begin{example}\label{Eg:m=135n=9q=16}
    Let $\de=135$, $\os=9$ and $q=16$, so that $F=\FF_{16}$ and $K=\FF_{2^{36}}$.  Then  $s_0=\gcd(15,135/9)=15>1$.  In particular, $s=s_0(q^\os-1)/(q-1)=q^\os-1$ by Lemma~\ref{L:s_0} so by Theorem~\ref{L:Chenclassesfinitefield} every element of $K^\times$ defines a distinct \emph{Chen-equivalence class}, yielding the maximum possible number $68,719,476,735$ of Chen equivalence classes.
These group, as in Theorem~\ref{them:countingskewconstacodes} into a large number of \emph{equivalence classes}.  To count them, we find the  factorization
$$2^{36}-1
=3^3\cdot5
\cdot7\cdot13
\cdot19\cdot37\cdot73\cdot109.$$
It is immediate that the lower bound of Corollary~\ref{cor:lowerbound} on the number of equivalence classes is $d(2^{36}-1)=512$.  Using SageMath \cite{sagemath} to evaluate the formula of Theorem~\ref{them:countingskewconstacodes} yields the actual value of $1,908,881,899$ equivalence classes,   most of which correspond to nonassociative ambient Petit algebras.

We now turn to the associative algebras, and show that some of their equivalence classes group into larger isometry classes.

Begin by noting that by Corollary~\ref{Cor:associativecase}, the  $q-1=15$ Chen-equivalence classes of \emph{associative} ambient algebras  group into
$$
\sum_{d|s_0} \frac{\varphi(d)}{o(2)_d} = \underbrace{\frac{1}{1}}_{d=1}+\underbrace{\frac{2}{2}}_{d=3}+\underbrace{\frac{4}{4}}_{d=5}+
    \underbrace{\frac{8}{4}}_{d=15}=5
    $$
distinct equivalence classes.

We now group them into isometry classes.
    The set $S$ of \eqref{E:S} is
    $$S=\{19, 28, 37, 46, 64, 73, 82, 91, 109, 118, 127\}.$$
The subset $\mathscr{P}$ of $(\ZZ/s_0\ZZ)^\times$ generated by $p=2$, as in Proposition~\ref{prop:proper equivalences}, is
$
\{1, 2, 4, 8\}$
and thus the elements $k\in S$ whose image in $(\ZZ/s_0\ZZ)^\times$ is not in $\mathscr{P}$ is
$$
\{28, 37, 73, 82, 118, 127\},
$$
whose residues modulo $s_0=15$ are all either $7$ or $13$.
These are therefore the elements that will yield isometries between certain non-equivalent algebras, namely, those corresponding to $a\in F^\times$ such that
$$
a^k \notin \bigcup_{\tau\in \Gal(F/\FF_p)}\tau(a)N_\de^\sigma(K^\times) = \{a, a^2, a^4, a^8\}.
$$
Suppose $F^\times = \langle w \rangle$, that is, $w$ is a primitive element.   Then one  sees directly that $\{a^7, a^{13}\}\cap \{a,a^2,a^4,a^8\}=\emptyset$ if and only if $a=w^\ell$ with $\gcd(\ell,15)=1$.

We conclude that the two equivalence classes of associative ambient algebras that are represented by the primitive element $w$ and its inverse yield isometric ambient algebras, that is,
$$
\Kone{w}\cong \Kone{w^{-1}}
$$
and that this Hamming-weight preserving isometry is not realized by any mere equivalence, that is, any map of the form $G_{\tau,\alpha,1}$.
\end{example}

\subsection{Conditions on \texorpdfstring{$a$}{a}}
The conditions of Proposition~\ref{prop:proper equivalences} are independent of $a\in F^\times$: they depend entirely on (complex) arithmetic conditions arising from the choices of
$$
\de, \os, q, p.
$$
The next step is to isolate such choices for which the rings $\Kone{a}$ have nontrivial ideals, meaning they provide collections of isometric, nonequivalent skew constacyclic codes.

Since $\os\mid\de$, the polynomial $f(t)=t^\de-a$, with $a\in F^\times$, is a central polynomial in $K[t;\sigma]$.  In fact,  it is given by $\tilde{f}(t^{\os})$ where
$$
\tilde{f}(x) = x^{\de/\os}-a
$$
is a polynomial in the commutative ring $K[x]$.
Consequently, if $\tilde{f}(x) \in K[x]$ is reducible, so is $f(t)\in K[t;\sigma]$.

Moreover, if we choose $a\notin N_\de^\sigma(K^\times)=(F^\times)^{\de/\os}$, then
any factorization of $\tilde{f}(x)$ (and $f(t)$) will not be a complete factorization into linear factors. This is now reduced to a factorization problem that can be solved in  SageMath, for example.

\begin{example}
In the setting and notation of Example~\ref{Eg:m=135n=9q=16}, let us choose  $k=28$ and $a=w^6$.  Then  $a^k=w^{168}=w^3\in F^\times$ and $G_{\id,1,k}:\Kone{a^k}
    \to \Kone{a}$ is an isomorphism.  Setting $f(t)=t^{135}-w^3$ and $x=t^9$, we have $\tilde{f}(x)=x^{15}-w^3$, which factors as
    $$
    (x^{15}-w^3)=(x^5 + w)  (x^5 + w^3 + w^2)  (x^5 + w^3 + w^2 + w) \in F[x].
    $$
This corresponds to the factorization
$$
f(t) = (t^{135}-w^3)=(t^{45} + w)  (t^{45} + w^3 + w^2)  (t^{45} + w^3 + w^2 + w) \in K[t;\sigma]
$$
For example, with $k=28$ the factor $t^{45}+w$ is sent to $(t^{28})^{45}+w \in \Kone{a}$, which reduces modulo $t^{135}-w^6$ as
$$
t^{135\times 9}t^45+w = (w^6)^9 t^{45}+w = w^9(t^{45}+w^7)=x^5 + w^3 + w + 1
$$
(using SageMath \cite{sagemath}), which is indeed one of the factors (up to an invertible scalar) of $t^{135}-a$, as expected.
\end{example}

In the following example, we illustrate how an isometry of the ambient Petit algebras can map one ideal to another, and thus send one skew $(\sigma,a)$ constacyclic code isometrically onto another.

\begin{example}
    In the same setting above, if we instead choose $a=w^5$ then $a^k=a$ and $f(t)=t^{135}-a$, then $\tilde{f}(x)=x^{15}-a$ and
    $$
    \tilde{f}(x)=(x^3 + w)  (x^3 + w + 1)  (x^3 + w^2 + w + 1)  (x^3 + w^3 + w + 1)  (x^3 + w^3 + w^2 + 1).
    $$
    This time, choosing the ideal generated by $x^3+w=t^{27}+2$, we find $$G_{\id,1,28}(t^{27}+w)=t^{27\times 28}+w = t^{5\times135+81}+w=(w^5)^5t^{81}+w = w^{10}(t^{81}+w^{6}).$$
Observe that the  degree of this polynomial is not equal to the degree of $f(t)$, but of course the corresponding codes, being isomorphic, are of the same dimension.  Namely, the code is the ideal generated by $t^{81}+w^{6}$, which is also generated by $\gcd(f(t),t^{81}+w^6)$. To compute this, we revert again to the central polynomial $x^9+w^6$, which factors using SageMath as
    $$
(x + w + 1)(x + w^3 + 1)(x + w^3 + w)(x^3 + w^2) (x^3 + w^3 + w + 1).
    $$
Thus $\gcd(\tilde{f}(x),x^9+w^6)=x^3 + w^3 + w + 1$, which yields $t^{27}+w^3+w+1$ as a generator of the corresponding  ideal.
\end{example}

These examples illustrate the abundance of interesting cases of isometries that are not mere equivalences and that have not been taken into account in the prevalent literature so far.

\section{Compliance with ethical standards}

\subsection*{Funding sources} This paper was partially written during the second author's stay as a Simons Professor in Residence at the University of Ottawa in July and August 2025. She gratefully acknowledges the support of CRM and the Simons Foundation. The first author's research is supported by NSERC Discovery Grant RGPIN-2025-05630.

\subsection*{Disclosure of potential conflicts of interest} The authors have no relevant financial or non-financial interests to disclose.

\subsection*{Data Availability Statement} No data were generated or used for this research.

%\bibliographystyle{amsalpha}
%\bibliography{References}
\providecommand{\bysame}{\leavevmode\hbox to3em{\hrulefill}\thinspace}
\providecommand{\MR}{\relax\ifhmode\unskip\space\fi MR }
\providecommand{\MRhref}[2]{%
  \href{http://www.ams.org/mathscinet-getitem?mr=#1}{#2}
}
\providecommand{\href}[2]{#2}

\end{document}